\documentclass{article}
\usepackage{makeidx}  
\usepackage[bookmarks]{hyperref}
\usepackage{
	amscd,
	amsmath,
	amssymb,
	amsthm,
	array,
	caption,
	chngpage,
	color,
	float,
	graphicx,
	latexsym,
	mathrsfs,
	multicol,
	multirow,
	pgfgantt,
	braket,
	sectsty,
	subcaption,
	tikz,
	url,
	verbatim,
	wrapfig
}

\usetikzlibrary{automata, positioning}
\usetikzlibrary{decorations.pathmorphing}
\tikzset{snake it/.style={decorate, decoration=snake}}

\newcommand{\abs}[1]{\lvert#1\rvert}

\renewcommand{\phi}{\varphi}

\newcommand{\CP}{\mathbb C\mathbb P}
\DeclareMathOperator{\GL}{GL}
\DeclareMathOperator{\PU}{PU}
\DeclareMathOperator{\U}{U}
\DeclareMathOperator{\Orth}{O}
\DeclareMathOperator{\SU}{SU}
\DeclareMathOperator{\SO}{SO}
\DeclareMathOperator{\SL}{SL}
\DeclareMathOperator{\PSL}{PSL}
\DeclareMathOperator{\perm}{perm}
\DeclareMathOperator{\Alt}{Alt}

\newtheorem{thm}{Theorem}
\newtheorem{cor}[thm]{Corollary}

\newtheorem{con}[thm]{Conjecture}

\newtheorem{df}[thm]{Definition}
\newtheorem{rem}[thm]{Remark}

\begin{document}
\title{Superposition as memory: unlocking quantum automatic complexity}
\author{Bj\o rn Kjos-Hanssen\thanks{
		This work
		was partially supported by
		a grant from the Simons Foundation (\#315188 to Bj\o rn Kjos-Hanssen).
		This material is based upon work supported by the National Science Foundation under Grant No.\ 1545707.
		}
}

\maketitle              

\begin{abstract}
	We define the semi-classical quantum automatic complexity $Q_{s}(x)$ of a word $x$ as the infimum in lexicographic order of those pairs of nonnegative integers $(n,q)$ such that there is a subgroup $G$ of the projective unitary group $\PU(n)$ with $\abs{G}\le q$ and
	with $U_0,U_1\in G$ such that, in terms of a standard basis $\{e_k\}$ and with $U_z=\prod_k U_{z(k)}$, we have $U_x e_1=e_2$ and $U_y e_1 \ne e_2$ for all $y\ne x$ with $|y|=|x|$.
	We show that $Q_s$ is unbounded and not constant for strings of a given length.
	In particular,
	\[
		Q_{s}(0^21^2)\le (2,12) < (3,1) \le Q_{s}(0^{60}1^{60})
	\]
	and $Q_s(0^{120})\le (2,121)$.
\end{abstract}

\section{Introduction}
	\paragraph{Quantum locks.}
		Imagine a lock with two states, ``locked'' and ``unlocked'', which may be manipulated using two operations, called 0 and 1. Moreover, the only way to (with certainty) unlock using four operations is to do them in the sequence 0011, i.e., $0^n1^n$ where $n=2$. In this scenario one might think that the lock needs to be in certain further states after each operation, so that there is some memory of what has been done so far. Here we show that this memory can be entirely encoded in superpositions of the two basic states ``locked'' and ``unlocked'', where, as dictated by quantum mechanics, the operations are given by unitary matrices. Moreover, we show using the Jordan--Schur lemma that a similar lock is not possible for $n=60$.
	\paragraph{Quantum security.} A problem with traditional padlocks is that a clever lock-breaker can seek to detect what internal state the lock is in part-way through the entering of the lock code. This problem disappears when the internal states are just superpositions of ``locked'' and ``unlocked''. Of course, there may be a positive probability that the system when observed part-way through the entering of the lack code is observed in the ``unlocked'' state. To remedy this, one could use a sequence of many locks, say $10^k$ for a suitable positive integer $k$, and add a third ``permanently locked'' state which, once reached, cannot be left. Then observing the lock will likely eventually result in landing in the permanently locked state. (In the case of the code $0^n1^n$, the permanently locked state might be implemented as having a probability related to the difference in the number of 0s and 1s entered in the code so far.)
	Note that in theory this is a purely quantum phenomenon: any simulation of the quantum device using classical hardware will be subject to the original problem that a lock-breaker may try to discern the internal states of the classical hardware.

	\paragraph{Quantum automata.}
		One of the fascinating aspects of quantum mechanics is how our understanding of states is enriched, with observable states, pure states, and mixed states. The notion of \emph{state} of a finite automaton begs for a generalization to the quantum realm.
		Indeed, quantum finite automata have been studied already \cite{Kondacs}.

		On the other hand automatic complexity introduced by Shallit and Wang \cite{MR1897300} has been related it to model selection in statistics \cite{fewpaths} 
		and to pseudorandomness generation with linear feedback shift registers \cite{Kjos-IEEE}.
		Other approaches to automatic complexity \cite{2017arXiv170109060S,MR3389924,MR2866563} yield better insight into infinite words.

		We shall consider complexity with respect to an arbitrary semigroup before considering the quantum case of the projective unitary group $\PU(n)$.

	\begin{df}\label{at-most}
		Let $T_X$ denote the set of all transformations of the set $X$; $T_X=\{f\mid f:X\to X\}$.
		The complexity of a string $x\in\{0,1\}^n$, $n\ge 0$, is the class of all semigroup actions $\varphi:G\to T_X$ for semigroups $G$ and sets $X$,
		with
		\begin{itemize}
			\item two\footnote{noncommuting, unless $x$ is a unary string like $0^n$} elements
		$\delta_0, \delta_1\in G$, inducing $\delta_y=\prod_{k=1}^{\abs{y}}\delta_{y(k)}$ for each $y\in\{0,1\}^n$, $y=y(1)\cdots y(\abs{y})$;
			\item an initial state $\alpha\in X$; and
			\item a final state $\omega\in X$,
		\end{itemize}
		such that $x$ is the only
		$y\in\{0,1\}^n$
		for which $\delta_y\alpha:=\varphi(\delta_y)\alpha=\omega$.

		In this case we say that $x$ has \emph{complexity at most} $\varphi$, or, if $\varphi$ is understood, \emph{complexity at most} $G$.
	\end{df}

\subsection{Quantum automatic complexity}
	Let $e^{(n)}_j$, $1\le j\le n$ be the standard basis for $\mathbb C^{n}$.
	Let $\U(n)$ be the group of unitary complex $n\times n$ matrices and let $\PU(n)$ be the projective unitary group.

	For $n\times n$ matrices $U_0$ and $U_1$ and a binary string $x$, we define
	\[
		U_x =
		\prod_{k=1}^{\abs{x}} U_{x(k)}.
	\]

	\begin{df}
		A \emph{quantum deterministic finite automaton} (quantum DFA) $M$ with $q$ states consists of an \emph{initial state} $\alpha\in{\CP}^q$,
		a \emph{final state} $\omega$,
		and $\delta_0, \delta_1\in \PU(q)$.
		We say that $M$ \emph{accepts} a word $x\in\{0,1\}^n$, $n\ge 0$ if
		\[
			\delta_x\alpha = \omega.
		\]
		Let $x\in\{0,1\}^n$, $n\ge 0$.
		The \emph{quantum automatic complexity of $x$}, $Q(x)$,
		is the least $q$ such that there exists a quantum DFA $M$ with $q$ states
		such that for all $y\in\{0,1\}^{n}$, $M$ accepts $y$ iff $y=x$.
		\begin{itemize}
			\item
			If we additionally require that $\delta_0,\delta_1$ generate a finite subgroup of $\PU(q)$, we obtain the \emph{finite quantum automatic complexity} $Q_f(x)$.
			\item
			If we require $\alpha=e_1$ and $\beta=e_2$ then we obtain \emph{semi-classical quantum automatic complexity} $Q_s$.
			\item
			If we require both of the extra requirements for $Q_s$ and $Q_f$, we get $Q_{sf}$.
		\end{itemize}
	\end{df}
	We can write $Q_s(x)\le (n,\infty)$ if $Q_s(x)\le n$, and $Q_s(x)\le (n,f)$ if $Q_{sf}(x)\le n$ as witnessed by a finite group of order $f$.
	This way we see $A_{\perm}$, the automatic complexity \cite{MR1897300} with the added restriction that the transition functions be permutations, as an upper bound for $n$ and a lower bound for $f$.
	Ordering these pairs lexicographically, we shall show that
	\[
		(3,121)\le Q_{s}(0^{60}1^{60})\le (121,\infty)
	\]
	assuming the following conjecture.
	\begin{con}\label{rush}
		$A_{\perm}(x)=\abs{x}+1$ for all $x$.
	\end{con}
	\begin{rem}
		We have verified Conjecture \ref{rush} for binary strings of length up to 9.
	\end{rem}
	\begin{thm}
			(1) If $Q_s(x)>(n,\infty)$ then $Q_s(x)\ge (n+1,A_{\perm}(x))$.
			(2) We always have $Q_s(x)\le (A_{\perm}(x),A_{\perm}(x))$.
	\end{thm}
	\begin{proof}
		For (1) we note that the quantum states can be considered as states, so that the Cayley graph of any group witnessing $Q_{sf}$ can be thought of as a witness for $A_{\perm}$.
		For (2) we note that we can restrict attention to only the states $e_j^{(q)}$, $1\le j\le q$, refusing to use superposition.
	\end{proof}
	Matrices $M$ of dimension $n\times n$ whose entries are 0 and 1, with exactly one 1 per column, act on $X=[n]=\{1,\dots,n\}$ by matrix multiplication in the following way:
	\[
		\varphi(M)(j)=k,\qquad\text{where } M e^{(n)}_j=e_k^{(n)}.
	\]
	If $M$ is additionally invertible then it thus induces an element of the symmetric group $S_n$ and belongs to $\Orth(n)$, the group of orthogonal matrices $M$ (satisfying $M^{-1}=M^T$).
	\begin{thm}
		For each string $x$, $Q_f(x)$ is finite.
	\end{thm}
	\begin{proof}
		By the embedding of $S_n$ into $\Orth(n)$ above, and then inclusion of $\Orth(n)$ into $\U(n)$ (simply because $U^{-1}=U^T$ for a real matrix $U$ implies $U^{-1}=U^{\dag}$),
		we have $Q_f(x)\le (A_{\perm}(x), A_{\perm}(x))\le (x+1, x+1)$.
	\end{proof}
	We also have $Q\le Q_f\le Q_{sf}$ and $Q\le Q_s\le Q_{sf}$. For our quantum lock analogy we want distinct initial and final states,
	whereas for automatic complexity $A(x)$ or $A_{\perm}(x)$ it is natural to not require that.

\section{Bounds on \texorpdfstring{$Q_s$}{Qs}}\label{generic}

	Arbitrary $q$-state DFA transition functions $\delta_0$, $\delta_1$ can be considered to belong to the matrix algebra $M_q$ of all $q\times q$ matrices.
	They are then exactly the matrices whose entries are 0 and 1, with exactly one 1 per column.
	And the nondeterministic case just corresponds to 0--1 valued matrices with \emph{not necessarily} exactly one 1 per column.
	Moving to arbitrary real matrices we can significantly reduce the required dimension, from $n/2+1$ \cite{Kjos-EJC} to 2, as we now explain.

	The following Theorem \ref{nov-24} indicates how any binary string can be encoded, in a sense, by two $2\times 2$ matrices.
	\begin{thm}\label{nov-24}
		For each binary string $x$ there exist $U_0,U_1\in \GL_2(\mathbb R)$
		such that
		$U_x e_1=e_2$ and for any $y\ne x$, $\abs{y}=\abs{x}$,
		$U_y e_1\ne e_2$.
	\end{thm}
	We omit the proof.
	\begin{thm}
		$Q(x)\le 2$ for all strings $x$.
	\end{thm}
	\begin{proof}
		It suffices to show that there is a free group generated by two unitary matrices.
		It is well-known
		\cite{MR3043070}
		that a generic pair of unitaries in U(2) generates a free group.
		Indeed, the existence of free subgroups of $\SO(3)$ (and hence its double cover
		$\SU(2)$) was already known to F. Hausdorff \cite{MR1511802}; see also \cite{MR0096732} and explicit examples in \cite{MR1284564}.
	\end{proof}
	Unfortunately, perhaps, free groups are incompatible with the ``semi-classical'' $e_1\mapsto e_2$ property in the following way:
	\begin{thm}\label{tempura}
		There is no word $x$ of length $>0$ and pair of unitary matrices $U_0$, $U_1$ such that
		$U_0$ and $U_1$ generate a free group and $U_xe_1=e_2$ in projective space.
	\end{thm}

\section{Unboundedness of \texorpdfstring{$Q_f$}{Qf}}
	As usual we denote by $H\trianglelefteq G$ that $H$ is a normal subgroup $G$, and by $[G:H]$ the index of $H$ in $G$.
	\begin{thm}[Jordan--Schur]\label{js}
		There is a function $f(n)$ such that given a finite group $G$ that is a subgroup of $M_n(\mathbb C)$, 
		there is an abelian subgroup $H\trianglelefteq G$ such that $[G:H]\le f(n)$.
	\end{thm}
	\begin{cor}\label{refute}
		For each $n$ there exists an $m$ such that for any $u,v\in U(n)$ which generate a finite group, we have $[u^m,v^m]=1$, i.e., $u^mv^m=v^mu^m$.
	\end{cor}
	\begin{proof}[Proof of Corollary from Theorem.]
		If $G$ is a finite group generated by $u$ and $v$, and $H$ a normal abelian subgroup of index $[G:H]=m$, then $u^mH=H$ and $v^mH=H$ (since any group element raised to the order of the group is the identity) and so $u^m$ and $v^m$ belong to $H$, hence, $H$ being abelian, they commute.
	\end{proof}
	\begin{thm}
		For each $n$ there is a binary string $x$ with $Q_f(x)>n$.
	\end{thm}
	\begin{proof}
		Let $x=0^m1^m$ where $m$ is as in Corollary \ref{refute}. Given $\delta_0, \delta_1\in \PU(n)=U(n)/U(1)$, choose $x$ and $y$ in $\U(n)$ such that
		$\delta_0=xU(1)$ and $\delta_1=yU(1)$. Then $\delta_{0^m1^m}=(xU(1))^m(yU(1))^m=x^my^mU(1)=y^mx^mU(1)=\delta_{1^m0^m}$.
	\end{proof}

	The extent to which $2\times 2$ matrices suffice for quantum automatic complexity is indicated in Table \ref{tempurum}.
	\begin{table}
	\centering
	\begin{tabular}{|c|c|c|}
		\hline
		& $e_1\mapsto e_2$ required & not required \\
		\hline
		finite group required & $\infty$ 
		& $\infty$ \\ 
		not required  & unknown 
		& $2$\\ 
		\hline
	\end{tabular}
	\caption{Supremum of quantum automatic complexity over all strings. In the case where $e_1\mapsto e_2$ is required (semi-classical quantum automatic complexity $Q_s$) but finiteness ($Q_f$) is not,
	we at least know that free groups cannot answer the question, by Theorem \ref{tempura}.}\label{tempurum}
	\end{table}
	\begin{thm}\label{gniniam}
		$Q_{sf}(0^{60}1^{60})>2$.
	\end{thm}
	\begin{proof}
		Note that we may assume our finite subgroups are primitive as there is no point in having a separate automaton disconnected from the witnessing one.
		Collins \cite{MR2334748} then shows that for $n=2$, the optimal value is $m=60$.
	\end{proof}
	On the other hand, we show below in Theorem \ref{maining} that $Q_{sf}(0^21^2)=2$, leaving a gap $(2,60)$ for the least $n$ such that $Q_{sf}(0^n1^n)>2$.
	The state of our knowledge of finiteness of quantum automatic complexity is given in Table \ref{tempurum}.
\section{Calculating \texorpdfstring{$Q_{s}(0011)\le (2,12)$}{Qs(0011)<=(2,12)}}\label{0011}
	The group $\mathrm{SU}(2)$ is the group of unit quaternions with the matrix representation \cite{qchu}
	\[
		\mathbf 1=\begin{bmatrix}1&0\\ 0&1\end{bmatrix},\qquad \mathbf i = \begin{bmatrix}i&0\\0&-i\end{bmatrix},\qquad \mathbf j = \begin{bmatrix}0&1\\-1&0\end{bmatrix},\qquad
		\mathbf k = \begin{bmatrix}0&i\\i&0 \end{bmatrix}
	\]
	where $i$ is the imaginary unit.
	We shall consider its order 24 subgroup the binary tetrahedral group
	\[
		\left\{\pm \mathbf 1, \pm \mathbf i, \pm \mathbf j, \pm \mathbf k, \frac12(\pm \mathbf 1\pm \mathbf i\pm \mathbf j\pm \mathbf k)\right\},
	\]
	also known by isomorphism as $\SL(2,3)$. Moreover we shall consider the order 12 quotient $\PSL(2,3)$ which is isomorphic to the alternating group $\Alt(4)$.
	\begin{thm}\label{warm-up}
		There exist $\mathbf a, \mathbf b\in\mathrm{SU}(2)$ such that
		\[
			\mathbf a\mathbf a\mathbf b\mathbf b\not\in\{\mathbf a,\mathbf b\}^4\setminus\{\mathbf a\mathbf a\mathbf b\mathbf b\}.
		\]
	\end{thm}
	\begin{proof} 
		It turns out we can use the binary tetrahedral group to realize 0011 within $\mathrm{SU}(2)$. Namely, let
		\[
			\mathbf a=\delta_0 = (\mathbf 1 + \mathbf i + \mathbf j - \mathbf k)/2,\qquad
			\mathbf b=\delta_1 = (\mathbf 1 + \mathbf i + \mathbf j + \mathbf k)/2
		\]
		in the quaternion representation,
		\[
			\mathbf a = \frac12\begin{bmatrix}1+i&1-i\\-i-1&1-i\end{bmatrix},\qquad \mathbf b = \frac12\begin{bmatrix}1+i&1+i\\i-1&1-i\end{bmatrix}.
		\]
		We can check
		that $\mathbf a\mathbf a\mathbf b\mathbf b=-\mathbf j$ is unique among 4-letter words in $\mathbf a, \mathbf b$.
	\end{proof}

	\begin{thm}
		For $x=0011$, there exist $\delta_0,\delta_1\in\SO(3)$ such that for all $y\in\{0,1\}^4$, $\delta_y=\delta_x$ iff $y=x$.
	\end{thm}
	\begin{proof}
		Another way to express $\mathbf a$ and $\mathbf b$ in Theorem \ref{warm-up} is as
		\[
			e^{i\varphi}\begin{bmatrix}e^{i\Psi}&0\\0&e^{-i\Psi}\end{bmatrix}\begin{bmatrix}\cos\theta&\sin\theta\\ -\sin\theta&\cos\theta\end{bmatrix}\begin{bmatrix}e^{i\Delta}&0\\0&e^{-i\Delta}\end{bmatrix}
		\]
		where $\phi=0$, $\theta=\pi/4$, and $\mathbf a$ has $(\Psi,\Delta)=(0,\pi/4)$ and $\mathbf b$ has $(\Psi,\Delta)=(\pi/4,0)$.
		Thus
		\[
			\mathbf a
			=\frac{1-i}{{2}}\begin{bmatrix}1&1\\ -1&1\end{bmatrix}\begin{bmatrix}i&0\\ 0&1\end{bmatrix}
			=\left(\frac1{\sqrt{2}}\begin{bmatrix}1&1\\ -1&1\end{bmatrix}\right)\left(\frac{1-i}{\sqrt{2}}\begin{bmatrix}i&0\\ 0&1\end{bmatrix}\right)
		\]
		is a product of two matrices $\mathbf r \mathbf s$ in $\SU(2)$. The first one, $\mathbf r$, corresponds \cite{MR0114876} to the $\SO(3)$ rotation
		\[
			\begin{bmatrix}0&0&-1\\0&1&0\\1&0&0\end{bmatrix}
		\]
		which is a 90-degree rotation in the $xz$-plane,
		and the second one, $\mathbf s$, to a 90-degree rotation
		\[
			\begin{bmatrix}0&-1&0\\1&0&0\\0&0&1\end{bmatrix}
		\] in the $xy$-plane in $\SO(3)$. We have
		\[
			\mathbf b
			=\frac{1-i}{{2}}\begin{bmatrix}i&0\\ 0&1\end{bmatrix}\begin{bmatrix}1&1\\ -1&1\end{bmatrix} = \mathbf s\mathbf r.
		\]
	\end{proof}

	One remaining wrinkle, taken care of in Theorem \ref{maining}, is to make sure the other words are not only distinct from $\mathbf a\mathbf a\mathbf b\mathbf b$, but map the start state to distinct vectors from what $\mathbf a\mathbf a\mathbf b\mathbf b$ does.

	\begin{thm}[well known]\label{known}
		The order 24 group $\SL(2, 3)$ is given by $a^3=b^3=c^2=abc$, or equivalently $a^3=b^3=abab$.
	\end{thm}
	\begin{thm}[{\cite{tribi}}]
	$\SL(2, 3)$ is isomorphic to the binary tetrahedral group, a subgroup of $\U(2)$.
	\end{thm}
	The group $\Alt(4)$ does serve as complexity bound for 0011. It is not a subgroup of $\U(2)$ \cite{tribi}, but:
	\begin{thm}\label{faith}
		There is a faithful, irreducible representation of $\Alt(4)\cong\PSL(2,3)$ as a subgroup of $\SL(2,3)$ of index 2 and as a subgroup of $\PU(2)$.
	\end{thm}
	\begin{proof}
		Let $a$ be as in Theorem \ref{known}.
		We define an equivalence relation $\equiv$ by $u\equiv v\iff u\in\{v,a^3v\}$.
		It is required to show that each element of our $\SL(2, 3)$ is equivalent to an element of $\Alt(4)$.
		This is done in detail in Figure \ref{orange}.
	\end{proof}
	\begin{figure}
		\begin{eqnarray}
			1,\\
			a,\\
			b,\\
			a^2 &=& {\color{red}bab},\\
			ab,\\
			ba,\\
			b^2 &=& {\color{red}aba},\\
			a^3 &=& b^3=baba=abab\equiv{\color{orange} 1},\\
			{\color{red}a^2b} &=& bab^2 ,\\ 
			{\color{red}ab^2} &=& {\color{blue}a^2ba},\\
			{\color{red}ba^2} &=& {\color{blue}b^2ab} = aba^2b,\\
			{\color{red}b^2a} &=& {\color{blue}aba^2} = ab^2ab,\\
			a^4 &=& ab^3=b^3a=a^2bab = ababa \equiv{\color{orange}a},\\
			a^3b &=& ba^3=b^4=abab^2\equiv{\color{orange}b},\\
			{\color{blue}a^2b^2} &=& a^3ba \equiv{\color{orange}ba},\\ 
			{\color{blue}ab^2a} &=& a^2ba^2,\\
			{\color{blue}ba^2b}&\equiv&{\color{orange}abba},\\
			{\color{blue}b^2a^2} &=& a^4b  = aba^3 = ab^4\equiv{\color{orange}ab},\\ 
			a^5 &=& a^2b^3  =ab^3a\equiv{\color{orange}a^2},\\ 
			a^3b^2 &\equiv&{\color{orange}b^2},\\ 
			a^2b^2a &\equiv&{\color{orange}baa},\\
			ab^2a^2 &\equiv&{\color{orange}aab},\\
			ba^2b^2&\equiv&{\color{orange}bba},\\
			b^2a^2b&\equiv&{\color{orange}abb}.
		\end{eqnarray}
		\caption{
			The 24 elements of $\SL(2, 3)$.
			All strings of length at most 2 are unique of their length.
			By symmetry, words of length 5 starting with $b$ are not written down.
		}\label{orange}
	\end{figure}
	The representation from Theorem \ref{faith} is used in the proof of Theorem \ref{maining}.
	\begin{thm}\label{maining}
		$Q_{sf}(0011)=2$.
	\end{thm}
	\begin{proof}
		Let $v=\begin{bmatrix}v_1\\ v_2\end{bmatrix}$ with $v_1,v_2\in\mathbb R$.
		Let
		\[
			E_0=\mathbf a = \frac12\begin{bmatrix}1+i&1-i\\-1-i&1-i\end{bmatrix},\qquad E_1=\mathbf b = \frac12\begin{bmatrix}1+i&1+i\\-1+i&1-i\end{bmatrix}.
		\]
		Let
		\[
			D = \begin{bmatrix} v_1 & -v_2\\ v_2 & v_1\end{bmatrix} = [v\mid E_{0011}v],\qquad C = \frac1{\sqrt{\det{D}}}D.
		\]
		Let
		\[
			U_j = C^{-1}E_jC,\qquad j\in\{0,1\}.
		\]
		Then it follows that
		\[
			CU_{0011}\begin{bmatrix}1\\ 0\end{bmatrix}
			=E_{0011}C\begin{bmatrix}1\\ 0\end{bmatrix}
			=E_{0011}\begin{bmatrix}v_1\\ v_2\end{bmatrix}
			=C\begin{bmatrix}0\\ 1\end{bmatrix}.
		\]
		Hence
		\[
			U_{0011}\begin{bmatrix}1\\ 0\end{bmatrix}=\begin{bmatrix}0\\ 1\end{bmatrix}.
		\]
		Since fortunately our $E_0$ and $E_1$ satisfy $E_{0011} = -\mathbf j$,
		$C$ is orthogonal, and in particular $C$ is unitary.
		If we now choose $v=\begin{bmatrix}1\\ 2\end{bmatrix}$, then $v$ is sufficiently generic that $U_ye_1\ne e_2$ as elements of $\CP^1$ for all $y\in\{0,1\}^4\setminus\{x\}$.
		We have verified as much with an Octave computation (see Figure \ref{cuboctahedron3} and Figure \ref{cuboctahedron2}). We have
		\begin{eqnarray*}
			C=\frac1{\sqrt{5}}\begin{bmatrix}1&-2\\ 2&1\end{bmatrix},\qquad
			C^{-1}=\frac1{\sqrt{5}}\begin{bmatrix}1&2\\ -2&1\end{bmatrix},
		\end{eqnarray*}
		\begin{eqnarray*}
			U_0 =\frac1{10}
			\begin{bmatrix}
			   5 + i&   5 + 7i\\
			  -5 + 7i&   5 - i
			\end{bmatrix}\qquad\text{and}\qquad
			 U_1 =\frac1{10}
			\begin{bmatrix}
			   5 - 7i&   5 + i\\
			  -5 + i&   5 + 7i
			\end{bmatrix}.
		\end{eqnarray*}
	\end{proof}
	\begin{rem}
		It is still a question what the nature of the ``complexity'' $Q$ and $Q_s$ are picking out is. If it is anything like $A_{\perm}$ it may be less than intuitive.
		However, there is some reason to believe that quantum automatic complexity is better than permutation automatic complexity at distinguishing strings of the same length.
		For permutation automatic complexity we do not know any example of strings of the same length having distinct complexity, but
		for quantum automatic complexity, $0^{60}1^{60}$ and $0^{120}$ form such an example. The latter has complexity at most 2
		whereas the former does not (Theorem \ref{gniniam}).
	\end{rem}
	\begin{figure}
		\makebox[\textwidth][c]{
		\begin{tikzpicture}[->,>=stealth',shorten >=1pt,auto,node distance=2cm,semithick]
			\tikzstyle{every state}=[draw=black,text=black]
			\node[state] (A) {\tiny$0$};
			\node[]      (B) [right of=A]{};
			\node[state] (C) [above of=B] {\tiny$\frac{16}{13}-\frac{15}{13}i$};
			\node[]      (D) [above of=A] {};
			\node[state] (E) [left of=D] {\tiny$-\frac9{13}+\frac{20}{13}i$};
			\node[state] (F) [above of=E] {\tiny $\frac{16}{13}+\frac{15}{13}i$};
			\node[state] (G) [above of=C] {\tiny $\frac{9}{37}+\frac{20}{37}i$};
			\node[]      (H) [above of=D] {};
			\node[state] (I) [above of=H] {\tiny $\frac34$};
			\node[state] (J) [right of=C] {\tiny $-\frac43$};
			\node[state] (K) [left of=E]  {\tiny $\infty$};
			\node        (L) [below of=A] {};
			\node[state] (M) [right of=L] {\tiny $-\frac{16}{37}-\frac{15}{37}i$};
			\node[state] (N) [left of=L] {\tiny $\frac9{37}-\frac{20}{37}i$};
			\node        (O) [right of=M] {};
			\node[state] (P) [below of=O] {\tiny $-\frac{9}{13}-\frac{20}{13}i$};
			\node        (Q) [left of=N] {};
			\node[state] (R) [below of=Q] {\tiny $-\frac{16}{37}+\frac{15}{37}i$};
			\path (A) edge [dashed] (E)
			(E) edge [dashed] (C)
			(C) edge [dashed] (A)
			(F) edge [color=red] (E)
			(C) edge [color=red] (G)
			(G) edge [dashed] (F)
			(F) edge [dashed] (I)
			(I) edge [dashed] (G)
			(G) edge [color=red] (J)
			(J) edge [color=red] (C)
			(E) edge [color=red] (K)
			(K) edge [color=red] (F)
			(M) edge [color=red] (N)
			(N) edge [color=red] (A)
			(A) edge [color=red] (M)
			(M) edge [dashed] (J)
			(K) edge [dashed] (N)
			(N) edge [dashed] (R)
			(R) edge [dashed] (K)
			(J) edge [dashed] (P)
			(P) edge [dashed] (M)
			(R) edge [color=red] (P)
			(P) edge [bend right=70, color=red] (I)
			(I) edge [bend right=70, color=red] (R);
		\end{tikzpicture}
		}
		\caption{
			Quantum complexity witness having the shape of a cuboctahedron.
			The label $\alpha$ represents the projective point $[1:\alpha]$.
			The initial state is $[1:0]$ denoted by 0 and the accept state is $[0:1]$ denoted by $\infty$.
			Dashed lines indicate multiplication by $U_0$.
			Solid lines indicate multiplication by $U_1$.
		}\label{cuboctahedron3}
	\end{figure}
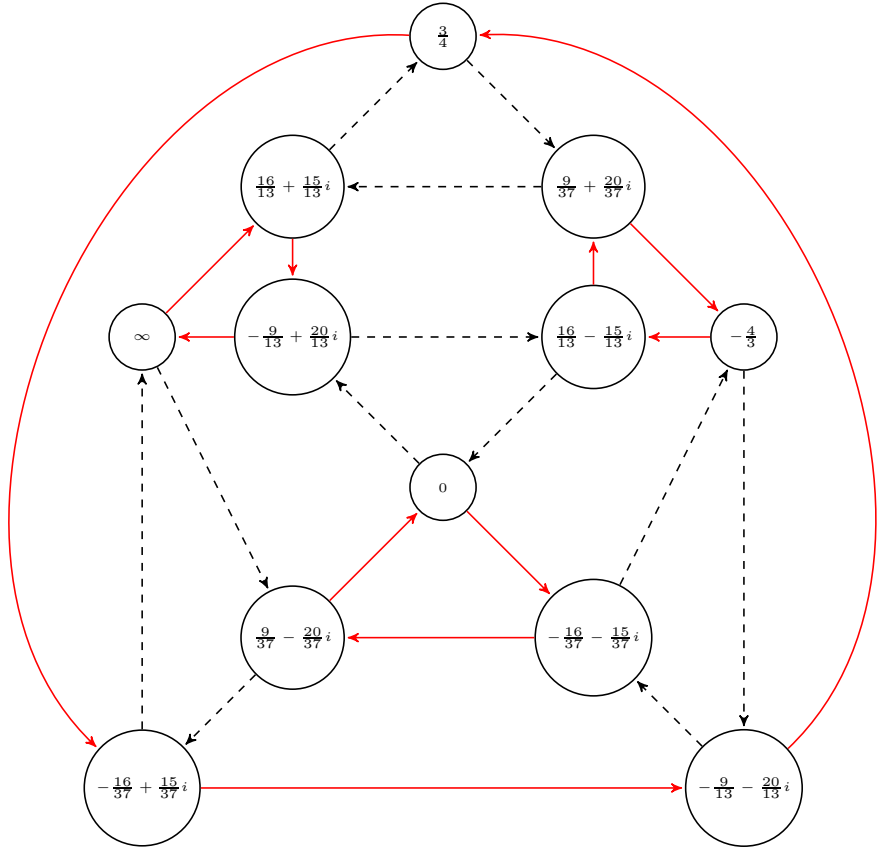
\begin{figure}
	\makebox[\textwidth][c]{
	\begin{tikzpicture}[->,>=stealth',shorten >=1pt,auto,node distance=2cm,semithick]
		\tikzstyle{every state}=[draw=black,text=black]
		\node[state] (A) {$I,(ab)^2,(ba)^2$};
		\node[]      (B) [right of=A]{};
		\node[state] (C) [above of=B] {$a^2$};
		\node[]      (D) [above of=A] {};
		\node[state] (E) [left of=D] {$a^4, b^3a, ab^3$};
		\node[state] (F) [above of=E] {$abaa$};
		\node[state] (G) [above of=C] {$bbab$};
		\node[]      (H) [above of=D] {};
		\node[state] (I) [above of=H] {$abba, baab$};
		\node[state] (J) [right of=C] {$bbaa$};
		\node[state] (K) [left of=E]  {$aabb$};
		\node        (L) [below of=A] {};
		\node[state] (M) [right of=L] {$ba^3,b^4,a^3b$};
		\node[state] (N) [left of=L] {$b^2$};
		\node        (O) [right of=M] {};
		\node[state] (P) [below of=O] {$babb$};
		\node        (Q) [left of=N] {};
		\node[state] (R) [below of=Q] {$aaba$};
		\path (A) edge [dashed] (E)
		(E) edge [dashed] (C)
		(C) edge [dashed] (A)
		(F) edge [color=red] (E)
		(C) edge [color=red] (G)
		(G) edge [dashed] (F)
		(F) edge [dashed] (I)
		(I) edge [dashed] (G)
		(G) edge [color=red] (J)
		(J) edge [color=red] (C)
		(E) edge [color=red] (K)
		(K) edge [color=red] (F)
		(M) edge [color=red] (N)
		(N) edge [color=red] (A)
		(A) edge [color=red] (M)
		(M) edge [dashed] (J)
		(K) edge [dashed] (N)
		(N) edge [dashed] (R)
		(R) edge [dashed] (K)
		(J) edge [dashed] (P)
		(P) edge [dashed] (M)
		(R) edge [color=red] (P)
		(P) edge [bend right=70, color=red] (I)
		(I) edge [bend right=70, color=red] (R);
	\end{tikzpicture}
	}
	\caption{
		Another view of the quantum complexity witness having the shape of a cuboctahedron of Figure \ref{cuboctahedron3}.
	}\label{cuboctahedron2}
\end{figure}
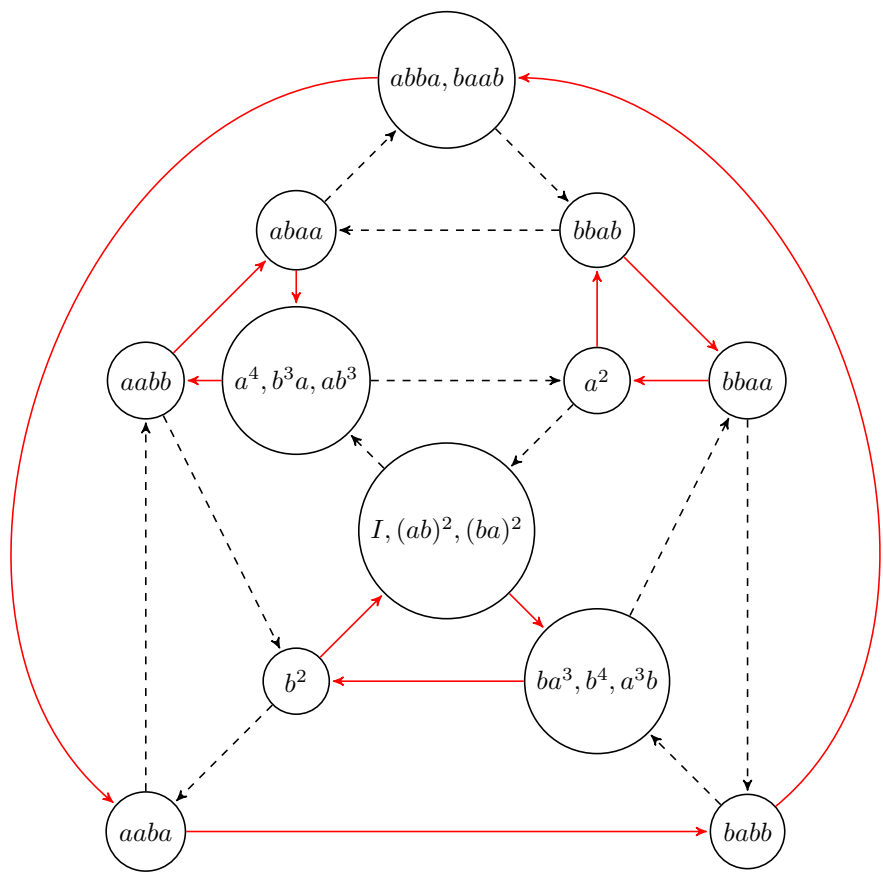
\bibliographystyle{IEEEtran}
\bibliography{UCNC17}
\end{document}